\documentclass[12pt]{article}
\usepackage[utf8]{inputenc}
\usepackage[margin=1in]{geometry}
\usepackage{amsmath, amssymb, amsthm}
\usepackage{color}
\usepackage{bbm} % e.g. \mathbbm{1}
\usepackage{natbib}
\usepackage{setspace}
\usepackage{bm}

\newtheorem{theorem}{Theorem}
\newtheorem{proposition}[theorem]{Proposition}
\newtheorem{lemma}[theorem]{Lemma}

\newtheorem{definition}{Definition}

\newtheorem{assumption}{Assumption}

\newtheorem*{theorem*}{Theorem}
\newtheorem*{proposition*}{Proposition}

\newcommand{\continuation}{??}

\newcommand{\commentout}[1]{}

\newcommand{\uk}{\underline{k}}
\newcommand{\ok}{\bar{k}}
\newcommand{\ux}{\underline{x}}
\newcommand{\ox}{\bar{x}}
\newcommand{\B}{\operatorname{Bin}}

\title{Optimal Decision Mechanisms for Committees: \\Acquitting the Guilty
\thanks{We thank Sarah Auster, Dan Bernhardt, Nina Bobkova, Martin Cripps, Duarte Goncalves, Johannes H\"orner, Philippe Jehiel, Aditya Kuvalekar, Daniel Kr\"ahmer, Stephan Lauermann, Benny Moldovanu,Axel Niemeyer, Justus Preusser, Ludvig Sinander, Vasiliki Skreta, Rani Spiegler as well as seminar participants at Naples, SAET 2023, SEA , SITE 2023, STF Oxford 2022, Toulouse, and UCL. This work is supported by ERC grant HWS - grant agreement No 101098319

}

}%: The Inferiority of Monotone Voting Rules under Strategic Voting}
\author{Deniz Kattwinkel \and Alexander Winter}
\date{\today 
}
 
\begin{document}
\onehalfspacing

\maketitle
\begin{abstract}

A group of privately informed agents chooses between two alternatives. How should the decision rule be designed if agents are known to be biased in favor of one of the options? We address this question by considering the Condorcet Jury Setting as a mechanism design problem. Applications include the optimal decision mechanisms for boards of directors, political committees, and trial juries. 

While we allow for any kind of mechanism, the optimal mechanism is a voting mechanism. In the terminology of the trial jury example: When jurors (agents) are more eager to convict than the lawmaker (principal), then the defendant should be convicted if and only  if neither too many nor too few jurors vote to convict. 

This kind of mechanism accords with a judicial procedure from ancient Jewish law.
\end{abstract}

\section{Introduction}
The Talmudic rules governing the ancient Jewish judicial system contain a peculiar detail that has puzzled religious and judicial scholars for centuries: If a defendant in a death penalty case was unanimously found guilty then he was to be acquitted. \footnote{For additional details and background information on this judicial procedure see \cite{bar2021talmudic} and \cite{glatt2013unanimous}.}

The question what sentencing rule should be used in a court relates to a central question of Economics that we study in this paper: how to aggregate knowledge that is dispersed among strategic agents. We derive the optimal decision mechanism in the Condorcet Jury model (\cite{condorcet}), the canonical theoretical model of dispersed information.

Applications include 
politicians having information about the benefits of a reform proposal, board members having knowledge about the benefit of a merger, and jurors having individual views on the guilt of a defendant after observing a trial.

Formal and informal institutions in society serve to  extract dispersed information and to map it into social outcomes. In our examples: parliamentary voting rules, board governance rules and jury decision procedures. 

Often, the preferences of the committee members are different from what would be socially or otherwise desirable.
In the examples, politicians could be more eager to change the status quo than would be the social optimum, board members could be biased towards a merger, and jurors could be more eager to convict the defendant than what would be optimal for society. In our setting these conflicts translate into different thresholds of doubt.\footnote{For instance, in the jury context, the lawmaker might want to convict only after being $90\%$ sure of the defendant's guilt whereas the jurors might want to convict after $50\%$. }

We find that a voting mechanism taking an interval form is optimal. In terms of the jury example: When jurors are more eager to convict than the lawmaker, then the defendant is convicted if and only if neither too many nor to few jurors vote to convict. This anti-unanimity accords with the decision rule described in the Talmud. 

In our model there is a binary social choice to be made and a number of agents privately receive conditionally independent binary signals about the state of the world. A principal can commit to a decision rule which defines a message set for each player and maps the joint messages of the agents into a lottery over choices.\footnote{By the revelation principle, this setting captures any potential mechanism, including dynamic mechanisms, where for example the principal consults the agents sequentially and uses the information he has gained for the interaction with subsequent agents.}

All players' preferences depend on the choice and the state. In one state of the world, decision $A$ is preferred, in the other state of the world decision $B$ is preferred. 
The players' preferences translate into a threshold of doubt: the minimal probability for the A-favouring state that would make the player want to choose $A$. We assume that the voters are homogeneous whereas the principal has a higher threshold of doubt.

We show that the optimal mechanism lets every agent vote and chooses the alternative according to the following decision rule: if the number of votes for $A$ lies in a certain interval, decision $A$ is implemented, otherwise decision $B$. This mechanism incentivizes the strategic agents to vote truthfully according to their signal. 

How can this non-monotone voting rule be the optimal mechanism?

The principal's first-best rule would choose $A$ whenever the number of signals in favour of $A$ induces a posterior belief about the $A$-state which exceeds the principal's threshold of doubt.

This mechanism does not incentivize the agents to vote truthfully: Suppose all other agents vote truthfully and consider an agent who received a signal in favour of $B$. The only case in which his vote is pivotal is when a vote for $A$ from him would just sway the posterior over the principal's threshold of doubt. If the agent's threshold of doubt is sufficiently lower than the principal's, in this situation the voter wants choice $A$ (even taking his own B-favouring signal into account). Since this is the only situation in which his vote makes a difference, he would rather vote untruthfully for $A$.\footnote{This insight is the same as in the non binding voting models:  \cite{morgan2008information}, \cite{levit2011nonbinding}, and \cite{battaglini2017public}}

The optimal mechanism addresses this problem by introducing an additional pivotal event (at the upper bound of the interval). To see how committing to $B$ for the very highest $A$-vote outcome gives the agents the incentive to vote truthfully consider first the agent-optimal rule which selects $A$ whenever the number of signals in favour of $A$ induces a posterior belief that exceeds the agents' threshold of doubt. This mechanism  incentivizes the agents to vote truthfully \citep{mclennan1998consequences}.

Consider the following numerical example: there are 9 agents, the agents' threshold of doubt is reached after 3 votes, the principal's after 5. The probability for an $A$-signal is $2/3$ in the $A$ state and $1/3$ in the $B$-state. 

The agent optimal mechanism would select $A$ after at least $3$ $A$-votes. The principal would like to lower the probability of selecting $A$ after $3$ and $4$ $A$-votes.

As the principal requires more $A$ votes than an agent to favour $A$ over $B$, the relevant incentive constraint is the one that ensures that an agent votes truthfully after a $B$ signal. 

Suppose we start lowering the probability after $3$ votes. This will reduce the incentives to report truthfully after a $B$ signal. There are two effects:

(i) If there are exactly two other agents that vote for $A$ and the agent votes untruthfully for $A$ the probability for selecting $A$ is reduced. So the damage from lying is reduced. (ii) In the other pivotal event, when there are exactly three other agents voting for $A$ and the agent reports truthfully for $B$, the probability of selecting $A$ (what the agent prefers in this situation) is lowered. So the benefit from voting truthfully is reduced. 

To counter this effect, suppose the we would now lower the probability of selecting $A$ after 9 $A$-votes. This reduces the benefit from lying after a $B$-signal since this event could only occur if the agent with a $B$-signal lies. 

When the probability of selecting $A$ after 3 votes is reduced to zero, we can start reducing the probability of selecting $A$ after 4 $A$-votes. This will have the same effect as before and we have to counter it by lowering the probability at the upper end of the interval.

If we lower the probability after 8 $A $-votes there are two countervailing effects on the incentives to vote truthfully after a $B$-signal: (i) If the other 8 agents voted for $A$ and the agent votes truthfully for $B$ the probability of selecting $A$ is decreased. The benefit from truthful voting is reduced. (ii) If only 7 other agents voted for $A$ and the agent would vote untruthfully for $A$, the probability of selecting $A$ would be reduced. The benefit from lying is decreased. 
Intuitively, the reason why the (ii) effect dominates the (i) effect is that in both states the probability of 7 $A$-votes is higher than that of 8 $A$-votes.\footnote{
7 and 8 are higher than the expected number of signals in state $A$ ($2/3*9=6$ ) and state $B$ ($1/3 *9=3$).}

This shows how we can successively reduce probability mass from the margins of the support interval. Starting with the agent-optimal mechanism we trade off the cost of reducing the probability of choosing $A$ after very high $A$-vote outcomes with the benefit of reducing the probability for selecting $A$ after medium $A$-votes outcomes. 

Our proof is based on linear programming and implies that in fact reducing the probability of selecting A at the margins of the intervals dominates any manipulation inside in terms of cost/benefit ratio. 
Theorem \ref{th:ivm} uses these arguments to show that an interval mechanism yields the principal the highest payoff.

 A principal with a very high threshold of doubt is even after a very high number  of $A$-votes not very sure that $A$ is the optimal decision. Therefore, 
 the costs of not choosing A after this high number of $A$-votes is rather small. Whereas the benefit of not choosing $A$ after a medium number of $A$-votes can be quite large for this kind of principal.
 In Theorem \ref{th:non-monotonic} we formalize this idea and show that for high enough thresholds of doubt an optimal mechanism will be non-monotonic.

\section{Related Literature}

The model we study in this paper can be traced back to the 
\cite{condorcet}. Condorcet studied the performance of a simple majority rule under the assumption of sincere voting. \cite{austen1996information} and \cite{feddersen1998convicting} analyzed this model in the presence of incentives and pointed out that sincere voting is often not an equilibrium. \cite{feddersen1998convicting} showed that requiring jurors to be unanimous can asymptotically increase rather than decrease the probability of convicting an innocent defendant. \cite{bouton2018get} argue that the unanimity rule is also inferior to a majority rule with veto power when agents sometimes have a private preference against one of the two options. \cite{chwe2010anonymous} shows that adding conflicts of interest between the agents can imply that non-monotonic voting rules are optimal. \cite{ali2019should} study the question whether a principal can benefit from banning deliberations between agents and find that doing so can be helpful if the voting rule is non-monotonic.

Papers that study the design of collective decision rules in related settings include \cite{wolinsky2002eliciting}, \cite{gershkov2017optimal}, \cite{chwe1999minority} and
        \cite{gershkov2009optimal}.

Details and additional background information on the judicial procedure described in the Talmud can be found in \cite{bar2021talmudic} and \cite{glatt2013unanimous}. See also \cite{levy2021maximum} and \cite{gunn2016too}.

\section{Model}
We consider the Condorcet Jury Setting\footnote{See for example \cite{feddersen1998convicting}.} with $n+1\ge 2$ agents $j$ and a principal $P$. 
Payoffs depend on a binary decision -- $A$ or $B$ -- and an unobserved state of the world $\omega \in \{\alpha,\beta\}$. In state $\omega$, alternative $A$ leads to payoff $V(\omega)$ to the principal and $U(\omega)$ to each agent. Payoffs from alternative $B$ are without loss of generality normalized to zero. Everyone strictly prefers alternative $A$ in state $\alpha$ and alternative $B$ in state $\beta$:
\begin{equation*}
    V(\alpha) > 0 > V(\beta),\quad U(\alpha) > 0 > U(\beta).
\end{equation*}
State $\alpha$ is realized with probability $\Pr(\alpha)\in(0,1)$. Agents receive conditionally independent private signals $s_j \in \{a,b\}$ about the state. In state $\omega$, the probability that an agent observes an $a$-signal is $p_\omega = \Pr(s_j=a|\omega) \in (0,1)$. Signals are informative in the sense that $a$-signals are more probable when the state is $\alpha$,\footnote{The direction of this inequality is without loss since signals can always be relabeled.}
\begin{equation*}\label{eq:informative}
    p_\alpha>p_\beta.
\end{equation*} 
Let $L(k)$ be the relative likelihood of state $\alpha$ when $k$ out of $n+1$ agents receive an $a$-signal:
\begin{equation*}
    L(k)= \frac{\Pr(\alpha)}{\Pr(\beta)}\left(\frac{p_\alpha}{p_\beta}\right)^k \left(\frac{1-p_\alpha}{1-p_\beta}\right)^{n+1-k}.
\end{equation*}
Since an additional $a$-signal is evidence in favor of state $\alpha$, $L(\cdot)$ is strictly increasing. If signals were public then the principal would weakly prefer decision $A$ if and only if 
\begin{equation*}
    L(k) \ge -\frac{V(\beta)}{V(\alpha)} =: t_P.
\end{equation*}
The value $t_P$ is called the principal's \textit{threshold of doubt}. Similarly, the threshold of doubt of the agents is given by
\begin{equation*}
    t_J = -\frac{U(\beta)}{U(\alpha)}.
\end{equation*}
We make the following additional assumptions: 
\begin{assumption}\label{as:ordering}
Agents are more partial toward $A$ than the principal: 
$    t_P\ge t_J.$
\end{assumption}
\begin{assumption}\label{as:nopartisans}
There are no partisans:
$
    L(0) < t_J, t_P < L(n+1).
$
\end{assumption}
\begin{assumption}\label{as:noindiff}
There are no indifferences:
$
    L^{-1}(t_J), L^{-1}(t_P) \in\mathbb{R}\setminus\mathbb{Z}.\footnote{When convenient we will interpret $L(\cdot)$ as a function $\mathbb R \to \mathbb R$.}
$\end{assumption}
Assumption \ref{as:ordering} says that the agents require a lesser level of confidence that the state is $\alpha$ in order to prefer alternative $A$ than the principal. This assumption is without loss of generality since the inequality could be reversed by relabeling the alternatives. Assumption \ref{as:nopartisans} says that conditioning on the realized number of $a$-signals, if no agent receives an $a$-signal then everyone strictly prefers alternative $A$, whereas if all agents receive an $a$-signals then everyone strictly prefers alternative $B$. This ensures that signals are informative enough to affect the preferences of all parties. Assumption \ref{as:noindiff} says that no realized number of $a$-signals would make the principal or the agents exactly indifferent between $A$ and $B$. This is a technical assumption made for mathematical convenience.

The principal designs and commits to a mechanism. Our main result characterizes the \textit{optimal} mechanism that maximizes the principal's expected payoff.

\section{Voting mechanisms}
The principal could design any static or sequential negotiation scheme: We allow for any kind of mechanism. Our first result is that no mechanism can make the principal better off than a simple class of mechanisms often encountered in practice: \textit{Voting mechanisms}. In a voting mechanism, each agent votes for one of the two alternatives and an alternative is chosen based on the vote tally. The result is not a direct consequence of the revelation principle but builds on the underlying symmetry of the setting. Recall that a direct mechanism is a mapping $\{a, b\}^{n+1}\to[0,1]$ that assigns to each type-profile a probability that $A$ is chosen. The revelation principle implies that it is without loss to restrict attention to Bayesian incentive-compatible (IC) direct mechanisms.\footnote{A direct mechanism is said to be IC if truthful signal-reporting is a Bayes-Nash equilibrium of the induced game.}
\begin{definition}
A direct mechanism that depends only on the number of reported $a$-signals is called a \textit{voting mechanism}. 
\end{definition}
Voting mechanisms can be written as $x:\{0,\cdots, n+1\}\to [0,1]$. An agent who reports an $a$-signal in a voting mechanism is said to \textit{vote} for $A$.\footnote{We use this terminology since an additional $a$-signal is evidence in favor of state $\alpha$ and all parties prefer decision $A$ in state $\alpha$}. If $k$ agents vote for $A$ then $A$ will be chosen with probability $x(k)$.
\begin{lemma}\label{la:votingmechs}
For any IC direct mechanism there exists an IC voting mechanism that leads to the same expected payoff for the principal. 
\end{lemma}
\begin{proposition}
There is an optimal mechanism that is a voting mechanism.
\end{proposition}
Since voting mechanisms are anonymous, IC reduces to two constraints: agents who observe an $a$-signal must find it worthwhile to vote for $A$ and agents who observe a $b$-signal must find it worthwhile to vote for $B$. Therefore the principal's problem reads:
\begin{align}%\label{eq:principalsproblem}
    \max\,\, &E[V(\tilde\omega) x(\tilde k)]\notag\\ \label{eq:IC-a}\tag{IC-a}
    \mbox{s.t.}\,\, &E[U(\tilde\omega) x(\tilde k)|a] \ge E[U(\tilde\omega) x(\tilde k-1)|a]\\ \label{eq:IC-b}\tag{IC-b}
    &E[U(\tilde\omega) x(\tilde k)|b] \ge E[U(\tilde\omega) x(\tilde k+1)|b],
\end{align}
where $\tilde k$ is the total number of $a$-signals among the agents.\footnote{Tildes indicate unrealized random variables.}

\section{First-best}
If signals were observable then the principal would choose $A$ over $B$ if and only if $L(k) > t_P$ holds for the realized number of $a$-signals. Let
\begin{equation*}
    k_P = \min\{k \in \{0,\dots,n+1\}: L(k) > t_P\}
\end{equation*}
be the minimal number of $a$-signals such that the principal would prefer $A$ over $B$. If incentives were not an issue then the principal would simply choose the \textit{first-best} or \textit{principal-preferred mechanism} 
\begin{equation}\label{eq:firstbestmech}
    x_P(k) = \begin{cases}
    0,& k< k_P\\
    1,& k\ge k_P.
    \end{cases}
\end{equation}
The number $k_P$ is called the \textit{principal-preferred cutoff}. $k_J$ and $x_J$ are defined identically, with $t_P$ replaced by $t_J$. We refer to $k_J$ and $x_J$ as the \textit{agent-preferred} cutoff and mechanism, respectively.
Assumptions \ref{as:ordering}, \ref{as:nopartisans} and \ref{as:noindiff} imply that 
\begin{equation}\label{eq:cutoffordering}
    0<k_J \le k_P \le n+1.
\end{equation}
If the agent- and principal-preferred cutoffs coincide then the principal's preferred decision coincides with the preferred decision given any signal profile. A \textit{conflict of interest} is said to exist if agents and principal disagree over their preferred decision given some signal profile, or, equivalently, if $k_J < k_P$.
\begin{definition}
There is said to be a conflict of interest if 
\begin{equation}\label{eq:conflict}
    k_J < k_P.
\end{equation}
\end{definition}
Suppose the agent-preferred mechanism is employed and some agent $j$ contemplates whether to vote for $A$ or $B$. If all other agents vote according to their signals then $j$ can influence the final decision if and only if exactly $k_J-1$ out of the $n$ other agents have received an $a$-signal (in this case $j$'s vote is said to be \textit{pivotal}). If $j$ has received an $a$-signal then there are exactly $k_J$ $a$-signals in total and to $j$ would strictly prefer to vote for $A$ (by definition of $k_J$). Conversely, if $j$ has received a $b$-signal then there are exactly $k_J-1$ $a$-signals in total and $j$ would strictly prefer to vote for $B$. Thus we have shown:\footnote{The fact that $x_J$ is IC can also be seen via a \cite{mclennan1998consequences}-type argument.}
\begin{lemma}\label{la:xJstrictIC}
The agent-preferred mechanism $x_J$ is strictly IC.
\end{lemma}
In particular, if there is no conflict of interest then $x_P$ is IC and achieves first-best. If, instead, there is a conflict of interest and the principal commits to $x_P$, then $j$'s vote is pivotal exactly when $k_P-1$ out of the $n$ other agents have received an $a$-signal. If $j$ has received an $a$-signal then he still strictly prefers to vote for $A$. But if he has received a $b$-signal then he now also strictly prefers to vote for $A$. Being pivotal after a $b$-signal means that exactly $k_P-1$ agents in total have received an $a$-signal. This  that $j$ would strictly prefer $A$. Since $k_P-1 \ge k_J$, he strictly prefers to vote for$A$. Hence, $x_P$ satisfies \ref{eq:IC-a}, but violates \ref{eq:IC-b} when there is a conflict of interest, implying that first-best cannot be achieved in this case. Together with Lemma \ref{la:xJstrictIC} this implies:
\begin{proposition}\label{prp:first-best}
First-best can be achieved if and only if there is no conflict of interest.
\end{proposition}

\section{Optimal mechanisms}
The previous section shows that first-best cannot be achieved when there is a conflict of interest. How should the principal design the mechanism in the latter case?

The majority of the existing literature studies voting mechanisms like $x_P$ and $x_J$ which are characterized by a single cutoff such that decision $A$ is implemented if and only if the number of $A$ votes exceeds this cutoff. Our main result shows that such single-cutoff mechanisms are in general non-optimal. Instead, optimal mechanisms are in general non-monotonic in the number of $A$-votes and take a specific form featuring two cutoffs rather than one. We call mechanisms of this type \textit{interval mechanisms}. Theorem \ref{th:ivm} below establishes that it is always optimal for the principal to employ an interval mechanism. Theorem \ref{th:non-monotonic} then shows that if the conflict of interest is large then an optimal interval mechanism must be non-monotonic unless it is constant.
\begin{definition}\label{def:intervalmechs}
A voting mechanism $x$ is said to be an interval mechanism if $x\equiv 0$ or there exist cutoffs $\uk\le\ok\in\{0,\cdots,n+1\}$ such that $\{k:x(k)>0\}=\{\uk,\dots,\ok\}$ and
\begin{equation*}
    x(k) =
    \begin{cases}1,&\uk<k<\ok \\
    0, & k< \uk \mbox{ or } k>\ok.
    \end{cases}
\end{equation*}
\end{definition}
A set of consecutive integers is called an \textit{interval} of integers. The \textit{boundary} of an interval of integers consists of its highest and lowest element and its \textit{interior} consists of all non-boundary elements. Given an interval mechanism $x$ as in Definition \ref{def:intervalmechs}, $\{\uk,\dots,\ok\}$ is called the \textit{implementation interval} %\textit{conviction interval} 
and to $\uk$ and $\ok$, respectively, as the lower and upper \textit{implementation boundaries}. 
\begin{theorem}\label{th:ivm}
The principal's problem is solved by an interval mechanism.
\end{theorem}

The proof of Theorem \ref{th:ivm} in Appendix \ref{apx:ivmproof} is based on a simple idea.
The agents have a lower threshold of doubt than the principal and are therefore ex-ante more partial toward $A$. Intuitively, the principal should therefore have no difficulty incentivizing agents to truthfully report ``optimistic'' signals, which would increase the principal's willingness to choose $A$. This suggests that the principal can ignore \ref{eq:IC-a} when designing the optimal mechanism. The resulting relaxed problem reads:
\begin{equation}\label{eq:relaxedproblem}
    \begin{aligned}
    \max\,\, &E[V(\tilde\omega) x(\tilde k)]\notag\\
    \mbox{s.t.}\,\, &E[U(\tilde\omega) x(\tilde k)|b] \ge E[U(\tilde\omega) x(\tilde k+1)|b],
\end{aligned}\tag{R}
\end{equation}
We then transform this problem into a linear program and apply Lagrangian relaxation techniques to merge \ref{eq:IC-b} with the principal's objective function. The key step is showing that this new ``virtual utility'' function crosses zero at most twice: first from below and then from above. This establishes that virtual utility is maximized by an interval mechanism. Finally, we verify that this interval mechanism satisfies \ref{eq:IC-a} and therefore also solves the original problem. 

\begin{definition}
Let $x$ be a voting mechanism. $x$ is said to be monotone if $x(k') \le x(k'')$ whenever $k' \le k''$. $x$ is said to be responsive if there exist $k'$ and $k''$ such that $x(k') \ne x(k'')$.
\end{definition}
The next result shows that a responsive optimal interval mechanism cannot be monotonic if the conflict of interest is sufficiently large.
\begin{theorem}\label{th:non-monotonic}
Assume there is a conflict of interest. There exists a threshold $\bar t_P < L(n+1)$ such that if
\begin{equation*}
    t_P \ge \bar t_P
\end{equation*}
then any responsive optimal interval mechanism must be non-monotonic.
\end{theorem}

\appendix
%\bibliography{acquit}

\section*{Appendix}
\section{Proof of Lemma \ref{la:votingmechs}}
\begin{proof}
Let $z:\{a,b\}^{n+1}\to[0,1]$ be an IC direct mechanism and let $\pi:\{1,\dots,n+1\} \to \{1,\dots,n+1\}$ be a permutation of the agents. Consider the mechanism $z_\pi(\cdot)$ defined by
\begin{equation*}
    (s_1,\dots,s_{n+1})\mapsto z(s_{\pi(1)},\dots,s_{\pi(n+1)}).
\end{equation*}
This is the same mechanism as before except that agents have been relabeled according to $\pi$. Relabeling does not break incentives since the agents are ex-ante identical and so $z_\pi(\cdot)$ is IC. Moreover, $z_\pi(\cdot)$ leads to the same expected payoff to the principal as $z(\cdot)$. This is because $E[V(\tilde\omega)z_\pi(\tilde s)]=E[V(\tilde\omega)E[z_\pi(\tilde s)|\tilde \omega]]$ and symmetry of the conditional type distribution implies that $E[z_\pi(\tilde s)|\tilde \omega] = E[z(\tilde s)|\tilde \omega]$.

Now define a new mechanism $z'(\cdot)$ by 
\begin{equation*}
    z'(\cdot) = \frac{1}{(n+1)!}\sum_{\pi} z_\pi(\cdot),
\end{equation*}
where the sum is over all permutations of $\{1,\dots,n+1\}$. Since the set of IC mechanisms is convex, $z'(\cdot)$ is IC and since each $z_\pi(\cdot)$ leads to the same expected payoff to the principal as $z(\cdot)$, so does $z'(\cdot)$. Finally, since $z'(\cdot)$ averages over all permutations of the agents it is an anonymous mechanism and hence $z'(s_1,\dots,s_{n+1})$ depends only on the number of $a$-signals among $s_1,\dots,s_{n+1}$.
\end{proof}

\section{The principal's problem as a linear program}\label{apx:LP}
To prove most of our results it is convenient to express the principal's problem as a linear program. Toward this aim we first introduce some notation. For any $k\in\{0,\dots,n+1\}$ define
\begin{equation*}
    \B_\omega(k,n) = \binom{n}{k}p_\omega^k(1-p_\omega)^{n-k}
\end{equation*}
with the convention that $\binom{n}{-1}=\binom{n}{n+1}=0$. Let
\begin{align*}
    v(k)&=\B_\beta(k,n+1)(L(k)-t_P),\\
    a(k)&=\B_\beta(k,n)p_\beta(L(k+1)-t_J),\\
    b(k)&= \B_\beta(k,n)(1-p_\beta)(L(k)-t_J).
\end{align*}
Then it is straightforward to show that the principal's problem reduces to the following linear program:
\begin{align}\notag{}
     \max_{0\le x(k) \le 1} & \sum_{k=0}^{n+1} v(k) x(k) \\
     s.t. \;\;\;\;\label{eq:IC-a'}\tag{IC-a'}
     & \sum_{k=0}^{n+1} [a(k)-a(k-1)]x(k)\le 0.\\\label{eq:IC-b'}\tag{IC-b'}
     & \sum_{k=0}^{n+1} [b(k)-b(k-1)]x(k)\ge 0.
\end{align}

\section{Proof of Theorem \ref{th:ivm}}\label{apx:ivmproof}
Firstly, if there is no conflict of interest then by Proposition \ref{prp:first-best}, first-best can be achieved. Since the first-best mechanism $x_P$ is an interval mechanism, there is nothing further to show in this case. Hence assume from now on that there is a conflict of interest.

In the notation of Appendix \ref{apx:LP}, the relaxed problem is equivalent to the following linear program:
\begin{equation}
\begin{aligned}
    \max_{0\le (x(k))_{k=0}^{n+1}\le 1} & \sum_{k=0}^{n+1} v(k) x(k)\\
     \mbox{s.t.} \;\;\;\;& \sum_{k=0}^{n+1} [b(k)-b(k-1)]x(k)\ge 0. 
\end{aligned}\tag{$R'$}
\end{equation}
Lemma \ref{la:relaxed-interval} below shows that this problem is solved by an interval mechanism. Lemma \ref{la:interval-IC-a} verifies that the resulting mechanism satisfies \ref{eq:IC-a}, and therefore also solves the principal's problem. The proofs of these lemmas rely on %Lemma \ref{la:xPIC-aIC-b} and 
three technical auxiliary results (Lemmas \ref{la:extreme}, \ref{la:exp+hyper}, and \ref{la:w}).

\begin{lemma}\label{la:relaxed-interval}
Assume there is a conflict of interest. The relaxed problem is solved either by $x^*\equiv 0$ or by an interval mechanism $x^*$ with cutoffs $\uk, \ok$ such that
\begin{equation*}
    k_J \le \uk \le k_P \le \ok.
\end{equation*}
Moreover, either (i) $x^*(\uk)=1$ and $\uk>k_J$, or (ii) $x^*(\ok)=1$.
\end{lemma}
\begin{proof}
The proof is based on a Lagrangian relaxation approach. The conflict of interest implies that a constraint qualification condition (Slater condition) holds for the relaxed problem, which implies that any solution to the relaxed must also solve the associated Lagrangian relaxation problem. Next, we show that a solution to the relaxed problem is necessarily an interval mechanism. Finally, we verify the conditions on the boundaries of the implementation interval by exploiting a classical result stating that an optimal solution of a linear program is attained at one of the extreme points of the feasible region. We proceed in five steps.
\begin{itemize}
    \item[] \textit{Step 1: There is a feasible mechanism $x$ with $\sum_{k=0}^{n+1} [b(k)-b(k-1)]x(k)> 0$:}
    
    The mechanism $x'\equiv 1$ that always chooses $A$, satisfies \ref{eq:IC-b} at equality. If we lower $x'(n+1)$ by some sufficiently small $\delta>0$ then the modified mechanism $x''$ satisfies
    \begin{equation*}
        \sum_{k=0}^{n+1} [b(k)-b(k-1)] x''(k)
        =-\delta[b(n+1)-b(N)]
        =\delta b(n) > 0,
    \end{equation*}
    since $b(n+1)=0$ and $b(n)=B_\beta(n, n)(1-p_\beta)(L(n)-t_J)>0$, where the latter inequality holds because $k_J<k_P\le n+1$ (conflict of interest).
        \item[] \textit{Step 2: For any solution $x^*$ to the relaxed problem, \ref{eq:IC-b} binds and there exists $\mu>0$ such that $x^*$ solves
        \begin{align}\label{eq:LR}
            \max_{0\le x(k)\le 1} \sum_{k=0}^{n+1} [v(k)+\mu (b(k)-b(k-1))]x(k):\tag{LR}
        \end{align}}
        
        By Step 1, the existence of a conflict of interest implies that the Slater condition is satisfied. Hence we can invoke Proposition 4 on page 348 of \cite{luenberger2015linear}, implying that if $x^*$ solves the relaxed problem then there exists there exists a Lagrange multiplier $\mu \ge 0$ such that $x^*$ also solves the Lagrangian relaxation problem \eqref{eq:LR}. Furthermore, $\mu\sum_{k=0}^{n+1}[b(k)-b(k-1)]x^*(k)=0$ (complementary slackness). If $\mu=0$ then the unique solution to \eqref{eq:LR} is $x_P$. But since there is a conflict of interest, $x_P$ violates \ref{eq:IC-b} (recall the discussion before Proposition \ref{prp:first-best})%(Lemma \ref{la:xPIC-aIC-b}) 
        and therefore cannot be a solution to the relaxed problem. It follows that $\mu > 0$ and hence, by complementary slackness, \ref{eq:IC-b} must bind.
        
        \item[]\textit{Step 3: For any $\mu > 0$ the function $v(k) + \mu (b(k) - b(k-1))$ crosses zero at most twice. If there are two crossings then the first is from below and the second from above.}
        
        First note that for $k=0,\dots,n+1$ it holds that
        \begin{align*}
            \binom{n}{k} 
            &= \frac{n+1-k}{n+1}\binom{n+1}{k} = \left(1-\frac{k}{n+1}\right)\binom{n+1}{k},\\
            \binom{n}{k-1} 
            &= \frac{k}{n+1} \binom{n+1}{k},
        \end{align*}
        so that 
        \begin{align*}
            &\B_\beta(k,n) = \frac{1}{1-p_\beta}
            \left(1-\frac{k}{n+1}\right)
            \B_\beta(k,n+1),\\
            &\B_\beta(k-1,n) = \frac{k}{p_\beta (n+1)}
            \B_\beta(k,n+1).
        \end{align*}
        Thus we have for $\mu>0$:
        \begin{align*}
            &v(k) + \mu [b(k) - b(k-1)]\\
            =&\B_\beta(k,n+1)(L(k) - t_P)\\
            &+\mu(1-p_\beta)[\B_\beta(k,n)(L(k) - t_J)-\B_\beta(k-1,n)(L(k-1) - t_J)]\\
            =&\B_\beta(k,n+1)
            \bigg\{(L(k) - t_P) \\
            &+\mu(1-p_\beta)\bigg[
            \frac{1}{1-p_\beta}
            \left(1-\frac{k}{n+1}\right)(L(k) - t_J)
            -
            \frac{k}{p_\beta (n+1)}(L(k-1) - t_J)\bigg]\bigg\}.
        \end{align*}
        Using that $L(k-1) = \frac{p_\beta}{1-p_\beta}\frac{1-p_\alpha}{p_\alpha}L(k)$, the above expression becomes
        \begin{align*}
            &\B_\beta(k,n+1)
            \bigg\{(L(k) - t_P)\\
            &+\mu\bigg[
            \left(1-\frac{k}{n+1}\right)(L(k) - t_J)-
            \frac{k}{(n+1)}\left(\frac{1-p_\alpha}{p_\alpha}L(k) - \frac{1-p_\beta}{p_\beta}t_J\right)\bigg]\bigg\}\\
            =&\B_\beta(k,n+1)
            \bigg\{(L(k) - t_P)\\
            &+\mu\bigg[
            \left(1-\frac{k}{p_\alpha(n+1)}\right)L(k)-
            \left(1-\frac{k}{p_\beta(n+1)}\right) t_J\bigg]\bigg\}\\
            =&\B_\beta(k,n+1)\phi(k),
        \end{align*}
        where 
        \begin{align*}
            \phi(k) 
            &= \left(
            1+\mu\left(1- \frac{k}{m_\alpha}\right)
            \right)L(k)
            - 
            \left(
            t_P + \mu t_J\left(1 - \frac{k}{m_\beta}\right)
            \right)\\
            &= 
            \frac{\mu}{m_\alpha}
            \left(
            \frac{1+\mu}{\mu}m_\alpha - k
            \right)L(k)
            - 
            \mu\frac{t_J}{m_\beta}
            \left(
            \frac{t_P}{t_J}\frac{m_\beta}{\mu}+m_\beta-k
            \right)
        \end{align*}
        and 
        \begin{equation*}
            m_\omega=p_\omega(n+1)
        \end{equation*} 
        is the expected number of $a$-signals if the state is $\omega$. 

        Since $\B_\beta(k, n+1)>0$ for $k=0,\dots,n+1$, $v(k) + \mu [b(k) - b(k-1)]$ has the same sign as $\phi(k)$ for any such $k$. We will proceed by considering $\phi(k)$ as a continuous function $\phi: \mathbb{R}\to\mathbb{R}$ and show that it has at most two zeros. Furthermore, if $\phi(k)$ has two zeros $k' < k''$ then $\phi(k) > 0$ for all $k'<k<k''$. This implies the claim.
        
        \textit{Case 1: $\frac{1+\mu}{\mu}m_\alpha = \frac{t_P}{t_J}\frac{m_\beta}{\mu}+m_\beta$.}
        
        Then 
        \begin{align*}
            \phi(k) 
            &=
            \frac{\mu}{m_\alpha}
            \left(
            \frac{1+\mu}{\mu}m_\alpha - k
            \right)
            \left(L(k)
            - 
            \frac{p_\alpha}{p_\beta}t_J
            \right), \quad k \in \mathbb{R}.
        \end{align*}
        
        This function equals zero exactly if $k = \frac{1+\mu}{\mu}m_\alpha$ or $L(k)-\frac{p_\alpha}{p_\beta}t_J$, is strictly positive for any $k$ that lies in the open interval between its zeros, and strictly negative for any $k$ that that lies outside the closed interval between its zeros.  
        
        \textit{Case 2: $\frac{1+\mu}{\mu}m_\alpha \ne \frac{t_P}{t_J}\frac{m_\beta}{\mu}+m_\beta$.}
        
        Then $\phi(\frac{1+\mu}{\mu}m_\alpha) \ne 0$ and 
        \begin{align*}
            \phi(k) 
            &=
            \frac{\mu}{m_\alpha}
            \left(
            \frac{1+\mu}{\mu}m_\alpha - k
            \right)
            \psi(k), \quad k \in \mathbb{R}\setminus\left\{\frac{1+\mu}{\mu}m_\alpha\right\},
        \end{align*}
        where
        \begin{align*}
            \psi(k)
            &= 
            L(k)
            - 
            \frac{p_\alpha}{p_\beta}t_J\frac{
            \frac{t_P}{t_J}\frac{m_\beta}{\mu}+m_\beta-k}{
            \frac{1+\mu}{\mu}m_\alpha - k}\\
            &= 
            L(k)
            -\left( 
            \frac{p_\alpha}{p_\beta}t_J
            +
            \frac{\frac{p_\alpha}{p_\beta}t_J\left(
            \frac{t_P}{t_J}\frac{m_\beta}{\mu}+m_\beta-\frac{1+\mu}{\mu}m_\alpha\right)}{
            \frac{1+\mu}{\mu}m_\alpha - k}\right),
            \quad k \in \mathbb{R}\setminus\left\{\frac{1+\mu}{\mu}m_\alpha\right\}.
        \end{align*}
        
        In the notation of Lemma \ref{la:exp+hyper}, $\psi(k)$ is therefore the difference between $L(k)$, an exponential function (with $Q = \left(\frac{1-p_\alpha}{1-p_\beta}\right)^{n+1} > 0$ and $R = \frac{p_\alpha/(1-p_\alpha)}{p_\beta/(1-p_\beta)}>1$) and a hyperbola (with $C = \frac{p_\alpha}{p_\beta}t_J >0$, $D=\frac{p_\alpha}{p_\beta}t_J\left(\frac{t_P}{t_J}\frac{m_\beta}{\mu}+m_\beta-\frac{1+\mu}{\mu}m_\alpha\right)\ne0$ and $E = \frac{1+\mu}{\mu}m_\alpha$).
        
        Note that, $\phi(E)\ne 0$ and $\phi(k)$ has the same sign as $\psi(k)$ on $(-\infty,E)$ and the same sign as $-\psi(k)$ on $(E, \infty)$. Applying Lemma \ref{la:exp+hyper} to $\psi(k)$, it now follows that $\phi(k)$ intersects zero at most twice on $\mathbb{R}$. Furthermore, if $\phi(k)$ intersects zero twice, then then the first intersection is from below and the second from above. 
        
        \item[]\textit{Step 4: If $x^*$ solves the relaxed problem then either $x^*\equiv 0$ or $x^*$ is a nonzero interval mechanism with thresholds satisfying $k_J \le \uk \le k_P \le \ok$.}
        
        A voting mechanism $x^*$ solves \eqref{eq:LR} if and only if $x^*(k)=1$ whenever $v(k)+\mu(b(k)-b(k-1)) > 0$ and $x^*(k)=0$ whenever $v(k)+\mu(b(k)-b(k-1)) < 0$. By Step 3, $v(k)+\mu(b(k)-b(k-1))$ crosses zero at most twice. If there are two crossings, then the first one is from below and the second one from above. Thus, either $v(k)+\mu(b(k)-b(k-1)) \le 0$ for all $k \in \{0,\dots,n+1\}$, or there exist $k', k'' \in \{0,\dots,n+1\}$, $k'\le k''$, such that $\{k \in \{0,\dots,n+1\}:v(k)+\mu(b(k)-b(k-1)) \ge 0\} = \{k',\dots,k''\}$ and $v(k)+\mu(b(k)-b(k-1)) > 0$ for all $k' < k < k''$. Hence any solution to \eqref{eq:LR} is an interval mechanism. But by Step 2 this means that any solution to the relaxed problem must also be an interval mechanism.
        
        It remains to show that a nonzero interval mechanism that solves the relaxed problem must have cutoffs $\uk, \ok$ satisfying $k_J \le \uk \le k_P \le \ok$. Let $x^*$ be a nonzero interval mechanism solving the relaxed problem. Toward a contradiction, suppose that $\ok < k_P$. Then $x\equiv 0$ would yield a strictly higher payoff, a contradiction. Hence $\ok \ge k_P$. Finally, suppose $\uk<k_J$. Since $\ok \ge k_P$, $x_J$ would yield a strictly higher payoff, a contradiction. Hence $\uk \ge k_J$.
        
        \item[]\textit{Step 5: If $x\equiv0$ does not solve the relaxed problem then there exists a solution $x^*$ to the relaxed problem that is a nonzero interval mechanism with $x^*(\uk)=1$ and $\uk>k_J$ or $x^*(\ok)=1$.}
        
        First note that the relaxed problem is a linear program with a nonempty and bounded feasible region. It is well-known (see for instance Corollaries 3.5, 3.6 and 3.10 in \cite{korte2018combinatorial}) that the optimal value of a linear program with a nonempty, bounded feasible region is attained at an extreme point of the feasible region. Hence there is a solution to the relaxed problem that is an extreme point of the polytope
        \begin{align*}
            \mathcal{X}_b = \left\{x \in \mathbb{R}^{n+2}: 0\le x(k) \le 1 \,\,\forall k,\,\, \sum_{k=0}^{n+1}x(k)(b(k)-b(k-1))\ge 0\right\} \subset \mathbb{R}^{n+2}.
        \end{align*}
        
        By Step 4, if $x^*$ solves the relaxed problem then either $x^*\equiv0$ or $x^*$ is an interval mechanism with $k_J \le \uk \le k_P \le \ok$. In particular, \textit{every} solution to the relaxed problem is an interval mechanism. Suppose that $x^*\equiv 0$ is not a solution. Since there is a solution that is an extreme point of $\mathcal{X}_b$ and every solution must be an interval mechanism, there exists a solution $x^*$ to the relaxed problem that is an interval mechanism and also an extreme point. By Lemma \ref{la:extreme}, $x^*$ has at most one entry strictly between 0 and 1. That is to say, $x^*(\uk) = 1$ or $x^*(\ok)=1$ holds.
        
        Finally, suppose that $x^*$ is a nonzero interval mechanism solving the relaxed problem and that $x^*(\uk)=1$. Toward a contradiction, suppose that $\uk=k_J$. If $k_P\le\ok<n+1$ or $x^*(\ok)<1$ then $x_J$ yields a strictly higher value of the relaxed problem, a contradiction. Thus $\ok=n+1$ and $x^*(n+1)=1$. That is, $x^*=x_J$. But we have assumed that there is a conflict of interest, and so lowering $x_J^*(\uk)$ slightly would strictly increase the value of the relaxed problem while leaving incentives intact (since $x_J$ is strictly IC by Lemma \ref{la:xJstrictIC}). Hence $x_J$ cannot be a solution to the relaxed problem, a contradiction. Thus, if $x^*(\uk)=1$ then it must hold that $\uk>k_J$.
\end{itemize}

\end{proof}

\begin{lemma}\label{la:interval-IC-a}
Let $x^*$ be a nonzero interval mechanism such that either (i) $x(\uk)=1$ and $\uk>k_J$ or (ii) $x(\ok)=1$ and $\uk\ge k_J$. If $x^*$ satisfies $IC\text{-}b$ at equality then it also satisfies $IC\text{-}a$.
\end{lemma}
\begin{proof}
First note that
\begin{equation*}
    a(k) = \frac{p_\beta}{1-p_\beta}b(k) w(k), \quad k=0,\dots,n+1,
\end{equation*}
where 
\begin{equation*}
    w(k) = \frac{L(k+1)-t_J}{L(k)-t_J}.
\end{equation*}
Recall that $L(k)\ne t_J$ for all integers so that $w(k)$ is well-defined. By Lemma \ref{la:w}, $w(k_J-1)<0$ and $w(k_J)>\dots>w(n+1)>0$. Define $\ux=x(\uk)$, $\ox = x(\ok)$.
\begin{itemize}
    \item[]\textit{Case 1: $\ox = 1$.}
    
    Since \ref{eq:IC-b} binds, it holds that
    \begin{align*}
        0&=\sum_{k=0}^{n+1} x(k)(b(k)-b(k-1))\\
        &=\sum_{k=\uk+1}^{\ok} (b(k)-b(k-1)) + \ux (b(\uk)-b(\uk-1))\\
        &=b(\ok)-(\ux b(\uk-1) + (1-\ux) b(\uk)).
    \end{align*}
    Regarding \ref{eq:IC-a}, we have
    \begin{align*}
        \sum_{k=0}^{n+1} x(k)(a(k)-a(k-1))
        &=a(\ok)-\left[\ux a(\uk-1) + (1-\ux) a(\uk)\right],
    \end{align*}
    which can be written as
    \begin{align}
        &\frac{p_\beta}{1-p_\beta}\left\{b(\ok)w(k)-\left[\ux a(\uk-1)w(\uk-1) + (1-\ux) a(\uk)w(\uk)\right]\right\}\notag{}\\
        =&\frac{p_\beta}{1-p_\beta}w(\ok)\left\{b(\ok)-\left[\ux b(\uk-1)\frac{w(\uk-1)}{w(\ok)} + (1-\ux) b(\uk)\frac{w(\uk)}{w(\ok)}\right]\right\}\notag{}\\
        =&-\frac{p_\beta}{1-p_\beta}w(\ok)\left[\ux b(\uk-1)\left(\frac{w(\uk-1)}{w(\ok)}-1\right) + (1-\ux) b(\uk)\left(\frac{w(\uk)}{w(\ok)}-1\right)\right]\label{eq:IC-aSlack1},
    \end{align}
    where the last equality uses that \ref{eq:IC-b} binds. Since $\ok \ge \uk \ge k_J$ it holds that $w(\ok)>0$. Furthermore, $b(\uk)\ge0$ and $\frac{w(\uk)}{w(\ok)}-1 \ge 0$.
    
    If $\uk=k_J$ then $b(\uk-1)\le 0$ and $\frac{w(\uk-1)}{w(\ok)} < 0$, so that \eqref{eq:IC-aSlack1} is non-positive and \ref{eq:IC-a} is satisfied. If instead $\uk>k_J$ then $b(\uk)\ge 0$ and $\frac{w(\uk-1)}{w(\ok)} > 1$. Then \eqref{eq:IC-aSlack1} is non-positive and IC-a is satisfied, as before.
    
    \item[]\textit{Case 2: $\ux=1$.}
    
    Similar to Case 1, \ref{eq:IC-b} being binding means that 
    \begin{equation*}
        0 = \left[\ox b(\ok) + (1-\ox) b(\ok-1)\right]-b(\uk-1).
    \end{equation*}
    Concerning \ref{eq:IC-a},
    \begin{align*}
        \sum_{k=0}^{n+1} x(k)(a(k)-a(k-1))
        &= \left[\ox a(\ok) + (1-\ox) a(\ok-1)\right]-a(\uk-1),
    \end{align*}
    which can be written as
    \begin{align}
        &\frac{p_\beta}{1-p_\beta}\left\{\left[\ox b(\ok)w(\ok) + (1-\ox) b(\ok-1)w(\ok-1)\right]-b(\uk-1)w(\uk-1)\right\}\notag{}\\
        =&\frac{p_\beta}{1-p_\beta}w(\uk-1)\left\{\left[\ox b(\ok)\frac{w(\ok)}{w(\uk-1)} + (1-\ox) b(\ok-1)\frac{w(\ok-1)}{w(\uk-1)}\right]-b(\uk-1)\right\}\notag{}\\
        =&\frac{p_\beta}{1-p_\beta}w(\uk-1)\left[\ox b(\ok)\left(\frac{w(\ok)}{w(\uk-1)}-1\right) + (1-\ox) b(\ok-1)\left(\frac{w(\ok-1)}{w(\uk-1)}-1\right)\right],\label{eq:IC-aSlack2}
    \end{align}
    where the last equality uses that \ref{eq:IC-b} binds. Since $\ok\ge\uk>k_J$, it holds that $w(\uk-1)>0$, $b(\ok)\ge 0$, $b(\uk-1)\ge 0$, $\frac{w(\ok)}{w(\uk-1)}>1$ and $\frac{w(\ok-1)}{w(\uk-1)}\ge1$. Thus, the last expression in \eqref{eq:IC-aSlack2} is non-negative and \ref{eq:IC-a} is satisfied.
\end{itemize}
\end{proof}

\section{Proof of Theorem \ref{th:non-monotonic}}
%The idea of the proof is as follows: if there is a conflict of interest, then show that there exists an IC, responsive, monotone interval mechanism $\hat x_J$ that achieves the highest expected payoff for the principal among all IC, responsive, monotone interval mechanisms. This mechanism satisfies \ref{eq:IC-b} at equality and \ref{eq:IC-a} strictly. We construct a specific ``deviation vector'' $\Delta \in \mathbb{R}^{n+2}$, such that $\hat x_J + \Delta$ becomes non-monotone but remains IC. The deviation vector lowers the probability of implementing $A$ after $k_J$ $a$-signals and after $n+1$ $a$-signals in exactly such a way that \ref{eq:IC-b} is unaffected. If the magnitude of the deviation is small enought then \ref{eq:IC-a} is still strictly satisfied for $\hat x_J + \Delta$. We then show that $\hat x_J + \Delta$ achieves a strictly higher expected payoff for the principal than $\hat x_J$, provided that $t_P$ lies above a certain threshold. This implies that in this case, a responsive optimal interval mechanism must necessarily be non-monotonic, proving the theorem.

Recall from Lemma \ref{la:xJstrictIC} that $x_J$ is strictly IC. Define the following interval mechanisms:
\begin{align*}
    \hat x_J (k)
    =
    \begin{cases}
        0, &k < k_J\\
        \frac{b(k_J)}{b(k_J)-b(k_J-1)} &k=k_J\\
        1, &k>k_J
    \end{cases}
    \quad\mbox{and}\quad
    \check x_J (k)
    =
    \begin{cases}
        0, &k < k_J-1\\
        \frac{a(k_J-1)}{a(k_J-1)-a(k_J-2)}, &k=k_J-1\\
        1, &k>k_J.
    \end{cases}
\end{align*}
$\hat x_J$ is a modified version of $x_J$ where $x_J(k_J)$ has been lowered until \ref{eq:IC-b} binds. In turn, $\check x_J$ is a modified version of $x_J$ where $x_J(k_J-1)$ has been increased until \ref{eq:IC-a} binds.

Lemma \ref{la:hatxJOptimalMonotone} below shows that if there is a conflict of interest and a monotone interval mechanism is optimal then it must be the case that either the zero-mechanism $x\equiv 0$ or $\hat x_J$ is optimal. Lemma \ref{la:hatxJOptimalMonotone} relies on Lemmas \ref{la:hat-xJ+check-xJ} and \ref{la:randomization}. We then show that $\hat x_J$ cannot be an optimal mechanism if $t_P$ is sufficiently large. This implies that a responsive optimal voting mechanism must necessarily be non-monotone when $t_P$ is sufficiently large. 

\begin{lemma}\label{la:hat-xJ+check-xJ}
The mechanisms $\hat x_J$ and $\check x_J$ are both IC. $\hat x_J$ satisfies \ref{eq:IC-a} strictly and \ref{eq:IC-b} at equality. If $k_J=1$ then $\check x_J\equiv 1$ and $\check x_J$ satisfies both \ref{eq:IC-a} and \ref{eq:IC-b} at equality. If $k_J>1$ then $\check x_J$ satisfies \ref{eq:IC-b} strictly at \ref{eq:IC-a} at equality. 
\end{lemma}

\begin{proof}
Incentive-compatibility of $\hat x_J$:
\begin{align*}
    \sum_{k=0}^{n+1}x(k)(b(k) - b(k-1)) &= \ux (b(k_J)-b(k_J-1) +b(n+1)-b(k_J) = 0,\\
    \sum_{k=0}^{n+1}x(k)(a(k) - a(k-1)) &= \ux (a(k_J)-a(k_J-1) +a(n+1)-a(k_J)\\ &= -((1-\ux) a(k_J)+\ux a(k_J-1)) < 0.
\end{align*}
Incentive-compatibility of $\check x_J$: If $k_J=1$ then $a(k_J-2)=0$ so that $\check x_J \equiv 1$ and there is nothing to show. If $k_J>1$ then with $\ux =  \frac{a(k_J-1)}{a(k_J-1)-a(k_J-2)} \in (0,1)$:
\begin{align*}
    \sum_{k=0}^{n+1}\check x_J(k)(b(k) - b(k-1))
    &= \ux (b(k_J-1)-b(k_J-2) +b(n+1)-b(k_J-1)\\ 
    &=-((1-\ux) b(k_J-1)+\ux b(k_J-2))
    > 0,\\
    \sum_{k=0}^{n+1}x(k)(a(k) - a(k-1)) &= \ux (a(k_J)-a(k_J-1) +a(n+1)-a(k_J) = 0.
\end{align*}
\end{proof}

\begin{lemma}\label{la:randomization}
Let $x$ be an IC responsive interval mechanism with $x(k)=1$ for some $k$. If $x$ is monotone then $\uk \in \{k_J-1, k_J\}$. If $\uk = k_J-1$ then $x$ is a randomization over $\check x_J$ and $x_J$; if $\uk = k_J$ then $x$ is a randomization over $\hat x_J$ and $x_J$. 
\end{lemma}
\begin{proof}
Let $x$ satisfy the assumptions of the lemma. Since $x(k)=1$ for some $k$ and $x$ is monotone it holds that $x(k)=1$ for all $k>\uk$. Let $\ux = x(\uk)$. 

We now show that $\uk \in \{k_J-1, k_J\}$. Suppose by contradiction that $\uk \le k_J-2$ (in particular, $k_J\ge 2$). Then 
\begin{align*}
    \sum_{k=0}^{n+1}x(k)(a(k)-a(k-1))
    &= -((1-\ux) a(\uk)+\ux a(\uk-1)) > 0,
\end{align*}
where the strict inequality follows directly if $k_J> 2$ (since then $a(\uk), a(\uk-1) <0$). In case $k_J =2$, responsiveness of $x$ implies $\ux <1$ and hence again the strict inequality (since $a(\uk)<0$). Thus, in either case, $x$ violates \ref{eq:IC-a}, contradiction.

Now suppose instead, suppose that $\uk \ge k_J+1$. Then 
\begin{align*}
    \sum_{k=0}^{n+1}x(k)(b(k)-b(k-1))
    &= -((1-\ux) b(\uk)+\ux b(\uk-1)) < 0,
\end{align*}
so \ref{eq:IC-b} is violated (because $b(\uk)\ge 0, b(\uk-1)>0$), contradiction.

We have shown that $\uk \in \{k_J-1, k_J\}$. To finish the proof, it suffices to show that $x(k_J-1) \le \check x(k_J-1)$ when $\uk = k_J-1$ and $x(k_J) \ge \hat x_J(k_J)$ when $\uk = k_J$.

\textit{Case: $\uk = k_J -1$.} If $k_J=1$ then $\check x(k_J-1)=1$ and there is nothing to show. Hence suppose $k_J\ge 2$. If $x(k_J-1) > \check x(k_J-1)$ then
\begin{align*}
    \sum_{k=0}^{n+1}x(k)(a(k) - a(k-1))  
    &= -((1-x(k_J-1)) a(\uk)+x(k_J-1)a(\uk-1))\\
    &>-((1-\check x(k_J-1)) a(\uk)+\check x(k_J-1)a(\uk-1)) = 0,
\end{align*}
$a(\uk)>0>a(\uk-1)$, which holds since $\uk=k_J-1$ and $k_J\ge 2$. So \ref{eq:IC-a} is violated, contradiction.

\textit{Case: $\uk=k_J$.} If $x(k_J) < \hat x_J(k_J)$ then 
\begin{align*}
    \sum_{k=0}^{n+1}x(k)(b(k) - b(k-1))  
    &= -((1-x(k_J)) b(\uk)+x(k_J)b(\uk-1))\\
    &<-((1-\check x(k_J)) b(\uk)+\check x(k_J)b(\uk-1)) = 0,
\end{align*}
since $b(\uk)>0>b(\uk-1)$. Hence \ref{eq:IC-b} is violated, contradiction.
\end{proof}

\begin{lemma}\label{la:hatxJOptimalMonotone}
Assume there is a conflict of interest. If a monotone interval mechanism is optimal then either $\hat x_J$ or $x\equiv 0$ is optimal.
\end{lemma}
\begin{proof}
Let there be a conflict of interest and suppose a monotone interval mechanism $x$ is optimal. Suppose the zero-mechanism $y\equiv 0$ is not optimal. Since $x_J$ yields a strictly higher expected payoff to the principal than $y\equiv 1$, $x$ has to be responsive. $x$ is IC, responsive, and monotone. Furthermore, the principal's expected payoff from $x$ must be positive, since $y\equiv 0$ is not an optimal mechanism. Therefore, since $x$ is optimal, there must exist $k$ such that $x(k)=1$. Let $\uk = \min\{k\in\{0,\dots,n+1\}:x(k)>0\}$. By Lemma \ref{la:randomization} and optimality of $x$, it holds that (i) $\uk=k_J$ (otherwise $x_J$ would yield a strictly higher expected payoff, because $v(k_J-1)<0$). Hence, again by Lemma \ref{la:randomization}, $x$ is a randomization over $x_J$ and $\hat x_J$. But $\hat x_J$ yields a strictly higher expected payoff than $x_J$ because $v(k_J)<0$ (conflict of interest). Thus, it must hold that $x=\hat x_J$.
\end{proof}

Now, to finish the proof of Theorem \ref{th:non-monotonic}, assume that there exists a conflict of interest and that the zero-mechanism is not optimal. By Lemma \ref{la:hat-xJ+check-xJ} the mechanism $\hat x_J$ satisfies \ref{eq:IC-b} at equality and \ref{eq:IC-a} strictly. By Lemma \ref{la:hatxJOptimalMonotone}, in order to show that an optimal responsive monotone interval mechanism must be non-monotone, it suffices to show that $\hat x_J$ is not an optimal mechanism. We will construct a deviation vector $\Delta$ such that $\hat x_J + \Delta$ yields a strictly higher payoff to the principal than $\hat x_J$, provided that $t_P$ is large enough.

Let $\delta > 0$ and define $\Delta \in \mathbb{R}^{n+2}$ by
\begin{align*}
    \Delta(k) = 
    \begin{cases}
        \frac{-\delta}{b(k)-b(k-1)}, &k=k_J\\
        \frac{\delta}{b(n+1)-b(n)}, &k=n+1\\
        0, &\mbox{otherwise}.
    \end{cases}
\end{align*}
Note that the above is well-defined because $k_J < k_P \le n+1$, in particular $k_J \ne n+1$. By construction of $\Delta$, 
\begin{equation*}
    \sum_{k=0}^{n+1}(\hat x(k)+\Delta(k))(b(k)-b(k-1)) = 0,
\end{equation*}
so \ref{eq:IC-b} holds at equality. Choose $\delta>0$ sufficiently small such that small such that $y = \hat x_J + \Delta$ is a feasible mechanism that satisfies \ref{eq:IC-a} strictly (this is possible because $\hat x_J$ satisfies \ref{eq:IC-a} strictly, see Lemma \ref{la:hat-xJ+check-xJ}). Then $y$ is an incentive-compatible mechanism. 

Using $y$ leads to a strictly higher expected payoff to the principal than using $\hat x_J$ if and only if 
\begin{align*}
    0
    &<\sum_{k=0}^{n+1}\Delta(k)v(k)\\
    &= \Delta(k_J)v(k_J) + \Delta(n+1)v(n+1)\\
    &= \delta\left(\frac{-v(k_J)}{b(k_J)-b(k_J-1)} + \frac{v(n+1)}{b(n+1)-b(n)}\right)\\
    &= -\delta\left(\frac{1}{\frac{b(k_J)}{v(k_J)}-\frac{b(k_J-1)}{v(k_J)}} + \frac{1}{\frac{b(n)}{v(n+1)}}\right)\\
    \Leftrightarrow
    0&> \frac{1}{\frac{b(k_J)}{v(k_J)}-\frac{b(k_J-1)}{v(k_J)}} + \frac{1}{\frac{b(n)}{v(n+1)}}.
\end{align*}
Note that $\frac{b(k_J)}{v(k_J)} < 0$, $\frac{- b(k_J-1)}{v(k_J)} < 0$ and $\frac{b(n)}{v(n+1)} > 0$. Hence the above condition rearranges to
\begin{align}\label{eq:*}\tag{$*$}
    0 
    &< \frac{b(n)}{v(n+1)} + \frac{b(k_J)}{v(k_J)}-\frac{b(k_J-1)}{v(k_J)}.
\end{align}
It holds that,
\begin{align*}
    (1-p_\beta)\frac{\B_\beta(n,n)}{\B_\beta(n+1,n+1)} &= (1-p_\beta)\frac{p_\beta^n}{p_\beta^{n+1}} = \frac{1-p_\beta}{p_\beta},\\
    (1-p_\beta)\frac{\B_\beta(k_J,n)}{\B_\beta(k_J,n+1)}
    &= \left(1-\frac{k_J}{n+1}\right)(1-p_\beta)\frac{{p_\beta}^k_J (1-p_\beta)^{n-k_J}}{{p_\beta}^k_J (1-p_\beta)^{n+1-k_J}} = \left(1-\frac{k_J}{n+1}\right),\\
    (1-p_\beta)\frac{\B_\beta(k_J-1,n)}{\B_\beta(k_J,n+1)}
    &= \frac{k_J}{n+1}(1-p_\beta)\frac{{p_\beta}^{k_J-1} (1-p_\beta)^{n+1-k_J}}{{p_\beta}^k_J (1-p_\beta)^{n+1-k_J}} = \frac{k_J}{n+1}\frac{1-p_\beta}{p_\beta}.
\end{align*}
Using the above and the definitions of $b$ and $v$, condition \eqref{eq:*} can be written as
\begin{align*}
    0 
    &< \frac{1-p_\beta}{p_\beta}\frac{L(n)-t_J}{L(n+1)-t_P} 
    + \left(1-\frac{k_J}{n+1}\right)\frac{L(k_J)-t_J}{L(k_J)-t_P}
    -\frac{k_J}{n+1}\frac{1-p_\beta}{p_\beta}\frac{L(k_J-1)-t_J}{L(k_J)-t_P}.
\end{align*}
Since there is a conflict of interest it holds that $L(k_J)<t_P$.
Multiplying by $L(k_J)-t_P <0$ and dividing by $\frac{1-p_\beta}{p_\beta}(L(n)-t_J)>0$ yields
\begin{align*}
    0 
    &> \frac{L(k_J)-t_P}{L(n+1)-t_P} 
    + \frac{p_\beta}{1-p_\beta}\left(1-\frac{k_J}{n+1}\right)\frac{L(k_J)-t_J}{L(n)-t_J}
    -\frac{k_J}{n+1}\frac{L(k_J-1)-t_J}{L(n)-t_J}\\
    &= \frac{L(k_J)-t_P}{L(n+1)-t_P} 
    + c',
\end{align*}
where
\begin{equation*}
    c' = \frac{p_\beta}{1-p_\beta}\left(1-\frac{k_J}{n+1}\right)\frac{L(k_J)-t_J}{L(n)-t_J}
    -\frac{k_J}{n+1}\frac{L(k_J-1)-t_J}{L(n)-t_J} >0.
\end{equation*}
But
\begin{align*}
    \frac{L(k_J)-t_P}{L(n+1)-t_P} 
    = \frac{L(k_J)-L(n+1)}{L(n+1)-t_P}+1.
\end{align*}
Using $t_P < L(n+1)$, the condition can be simplified to
\begin{equation*}
    t_P > L(n+1)-\frac{L(n+1)-L(k_J)}{1+c'} =: \bar t_P.
\end{equation*}
\qed

\section{Technical Results}
\begin{lemma}\label{la:extreme}
If $x \in \mathbb{R}^{n+2}$ is an extreme point of the polytope
\begin{align*}
    \mathcal{X}_b = \left\{x \in \mathbb{R}^{n+2}: 0\le x(k) \le 1\quad \forall k,\quad \sum_{k=0}^{n+1}x(k)(b(k)-b(k-1))\ge 0\right\} \subset \mathbb{R}^{n+2}.
\end{align*}
then there exists at most one $k' \in \{0,\dots,n+1\}$ such that $0<x(k')<1$.
\end{lemma}
\begin{proof}
By a theorem of Hoffman and Kruskal (see for example Proposition 3.9 in \cite{korte2018combinatorial}), a point in $\mathcal{X}_b$ is an extreme point if and only if it is the unique solution to a subsystem of the system of equations
        \begin{align*}
            x(k) &= 0, \quad k=0,\dots,n+1\\
            x(k) &= 1, \quad k=0,\dots,n+1\\
            \sum_{k=0}^{n+1}x(k)(b(k)-b(k-1))&=0.
        \end{align*}
        Let $x$ be an extreme point of $\mathcal{X}_b$. We claim that $x$ cannot have more than one entry that lies strictly between 0 and 1. We start with two observations. 
        
        First, if $x$ has an entry strictly between 0 and 1, then it is not the unique solution of any subset of the first $2(n+1)$ equations above. Therefore $\sum_{k=0}^{n+1}x(k)(b(k)-b(k-1))=0$ must hold. 
        
        Second, if $k'$ is such that $0<x(k')<1$, then $b(k')-b(k'-1)\ne 0$ must hold: Otherwise, it would be possible to express $x$ as 
        \begin{equation*}
            x = x(k') x' + (1-x(k')) x'',
        \end{equation*}
        where $x'(k')=1$, $x''(k')=0$, and $x'(k)=x''(k)=x(k)$ for all other $k$. This would contradict the assumption that $x$ is an extreme point, because $x', x'' \in \mathcal{X}_b$.
        
        Now, suppose by contradiction that there exist $k'<k''$ such that $0<x(k'), x(k'')<1$. Let $\delta>0$, and define $\Delta\in\mathbb{R}^{n+2}$ by
        \begin{align*}
            \Delta(k) = 
            \begin{cases}
                \frac{\delta}{b(k')-b(k'-1)}, & k=k'\\
                \frac{-\delta}{b(k')-b(k'-1)}, &k=k''\\
                0, &\text{otherwise}.
            \end{cases}
        \end{align*}
        Let $y = x + \Delta$ and $z = x-\Delta$. Then
        \begin{equation*}
            x = \frac{1}{2}y + \frac{1}{2}z.
        \end{equation*}
        But if $\delta>0$ is sufficiently small, then $y, z \in \mathcal{X}_b$, contradicting our assumption that $x$ is an extreme point. Hence, every extreme point of $\mathcal{X}_b$ has at most one entry strictly between 0 and 1.
\end{proof}

\begin{lemma}\label{la:exp+hyper}
Let $Q, R, C, D, E \in \mathbb{R}$, where $Q>0, R>1, C>0, D\ne 0$. For $t \in \mathbb{R}$ define
\begin{align*}
    e(t) &= QR^t,\\
    h(t) &= C + \frac{D}{t-E} \quad (t \ne E).
\end{align*}
The exponential function $e(t)$ intersects\footnote{We say that $e(t)$ \textit{intersects $h(t)$ from below} at $\hat t \in \mathbb{R}\setminus\{E\}$ if $e(\hat t) = h(\hat t)$ and there exists $\varepsilon > 0$ such that $e(t)<h(t)$ for all $t \in (\hat t - \varepsilon, \hat t)$ and $e(t)>h(t)$ for all $t \in (\hat t, \hat t+\varepsilon)$. An \textit{intersection from above} is defined analogously. We say that $e(t)$ \textit{intersects $h(t)$} at $\hat t \in \mathbb{R}\setminus\{E\}$ (without qualifiers) if it does so either from below or from above.} the hyperbola $h(t)$ at most twice on $\mathbb{R}-\{E\}$. Furthermore, if $D>0$ then $e(t)$ intersects $h(t)$ exactly once on $(-\infty,E)$ (from below) and exactly once on $(E,\infty)$ (from below). If $D<0$ then $e(t)$ intersects $h(t)$ on at most one of the intervals $(-\infty,E)$ and $(E,\infty)$. If $e(t)$ intersects $h(t)$ on $(-\infty,E)$ then it does so twice, first from below and then from above. If $e(t)$ intersects $h(t)$ on $(E,\infty)$ then it does so twice, first from above and then from below.
\end{lemma}
\begin{proof}$\;$

Since $Q>0$, we can divide by $Q$ if necessary. Hence we can without loss of generality assume that $Q=1$.
\begin{itemize}
    \item []\textit{Case 1: $D>0$.} 
    
    Then, $h(t)$ restricted to either $(-\infty,E)$ or $(E,\infty)$ is strictly decreasing. It follows that $e(t)-h(t)$ is strictly increasing when restricted to either of the intervals. Thus $e(t)-h(t)$ intersects zero at most once in each interval and any intersection must be from below. 
    
    Now note that $e(-\infty)=0<C=h(-\infty)$ and $e(E-)=e(E)>-\infty=h(E-)$. Hence, by continuity, there must be at least one intersection on $(-\infty,E)$. Finally, since $h(E+) = \infty>e(E)=e(E+)$ and $h(+\infty)=C<\infty=e(\infty)$, there has to be at least one intersection on $(E,\infty)$.
    
    \item []\textit{Case 2: $D<0$.}
    \begin{itemize}
        \item []\textit{Step 1: Either all intersections of $e(t)$ with $h(t)$ lie in $(-\infty,E)$ or all lie in $(E,\infty)$.}
    
        The hyperbola $h(t)$ lies strictly above $C$ on $(-\infty,E)$ and strictly below $C$ on $(E,\infty)$. Hence, if $e(t)$ intersects $h(t)$ on $(-\infty,E)$ then $e(t) > C$ on $(E,\infty)$ (since $e(t)$ is strictly increasing) and so there cannot be an additional intersection on $(E,\infty)$.
        
        \item []\textit{Step 2: $e(t)$ intersects $h(t)$ at most twice on $(E,\infty)$. If there are two intersections then the first one is from above and the second one from below.}
        
        On $(E,\infty)$, $h(t)$ is strictly concave, hence $e(t)-h(t)$ is strictly convex. Therefore, $e(t)-h(t)$ intersects zero at most twice and $e(t)-h(t)<0$ for any $t$ lying strictly between its zeros and $e(t)-h(t)>0$ for any $t$ lying outside the closed interval between its zeros. In particular, if $e(t)-h(t)$ intersects zero twice on $(E,\infty)$ then the first intersection is from above and the second from below.
        
        \item []\textit{Step 3: $e(t)$ intersects $h(t)$ at most twice on $(-\infty,E)$. If there are two intersections then the first one is from below and the second one from above.}
        
        Since $C>0$ it holds that $e(t), h(t) > 0$ for all $t \in (-\infty,E)$ and
        \begin{align*}
        &e(t)\ge h(t)\\
        \Leftrightarrow 
        &\log e(t) \ge \log h(t)\\
        \Leftrightarrow 
        &t \log R - \log \left(C + \frac{D}{t-E}\right) \ge 0.
        \end{align*}
        But the function on the LHS is strictly concave (see below) and therefore crosses zero at most twice. Furthermore, if there are two intersections, then the first is from below and the second from above. 
        
        Strict concavity of $t \log R - \log \left(C + \frac{D}{t-E}\right)$ follows by differentiating twice with respect to $t$, yielding
        \begin{align*}
            -\frac{2 C D(t-E)+D^2}{(E-t)^2 (C (t-E)+D)^2}<0,
        \end{align*}
        where we have used that $D<0$ and $(t-E)<0$. 
    \end{itemize}
\end{itemize}
\end{proof}

\begin{lemma}\label{la:w}
Let $\hat k_J = L^{-1}(t_J) \in \mathbb{R}$. The function $w:\mathbb{R}\setminus\{\hat k_J\}\to \mathbb{R}$,
\begin{equation*}
    w(k) = \frac{L(k+1)-t_J}{L(k)-t_J}
\end{equation*}
is strictly decreasing on $(-\infty, \hat k_J)$ and on $(\hat k_J, \infty)$. Furthermore,
\begin{align*}
    &w(k) > 0, \quad\quad\quad k\in\{0,\dots,n+1\}\setminus\{k_J-1\},\\
    &w(k_J-1)<0.
\end{align*}
\end{lemma}
\begin{proof}
For $k \in \mathbb{R}\setminus\{k_J\}$, let $w'(k)$ be the derivative of $w(k)$. Then
\begin{align*}
    w'(k)
    &=\frac{L'(k+1)(L(k)-t_J)-L'(k)(L(k+1)-t_J)}{(L(k)-t_J)^2}\\
    &= \frac{\rho}{(L(k)-t_J)^2} [L(k+1)(L(k)-t_J)-L(k)(L(k+1)-t_J)]\\
    &=\frac{\rho}{(L(k)-t_J)^2} [L(k)t_J-L(k+1)t_J]< 0,
\end{align*}
using that $    L'(k) =\rho L(K)$ with $\rho=\log\left(\frac{p_\alpha(1-p_\beta)}{p_\beta(1-p_\alpha)}\right)>0.$ Hence $w(k)$ is strictly decreasing on $(-\infty, k_J)$ and on $(k_J, \infty)$.

Since $L(k)$ is strictly increasing on $\mathbb{R}$, it holds that $L(k)<t_J$ for $k < \hat k_J$ and $L(k) > t_J$ for $k > \hat k_J$. Thus 
\begin{align*}
    w(k) 
    = \frac{L(k+1) - t_J}{L(k)-t_J}
    \quad
    \begin{cases}
        > 0, &k \in (-\infty, \hat k_J - 1) \cup (\hat k_J, \infty)\\
        < 0, &k \in (\hat k_J - 1, k_J).
    \end{cases}
\end{align*}
By Assumption \ref{as:noindiff}, $\hat k_J$ is not an integer. Thus $k_J-2<\hat k_J-1<k_J-1 < \hat k_J < k_J$ and therefore  
\begin{align*}
    w(k)
    \quad
    \begin{cases}
        > 0, &k\in\{0,\dots,n+1\}\setminus\{k_J-1\}\\
        < 0, &k=k_J-1.
    \end{cases}
\end{align*}
\end{proof}

\bibliography{acquit}

\end{document}